\documentclass[pra,twocolumn,superscriptaddress,nofootinbib,notitlepage]{revtex4}

\usepackage{amsthm}
\usepackage{bm}
\newtheorem{theorem}{Theorem}
\newtheorem{lemma}[theorem]{Lemma}
\theoremstyle{definition}

\newtheorem{remark}{Remark}
\usepackage{amsmath}
\usepackage{amssymb,amscd}
\usepackage[compatibility=false]{caption}
\usepackage{graphicx}
\usepackage{braket}
\usepackage{dsfont}
\usepackage{color}
\usepackage{tikz}
\usetikzlibrary{arrows,positioning}
\usetikzlibrary{shapes.geometric}
\usetikzlibrary{patterns}
\usepackage{mathrsfs}
\usepackage{algorithm}
\usepackage{algorithmicx}
\usepackage{algpseudocode}
\usepackage{hyperref}
\usepackage{mathtools}

\algdef{SE}[DOWHILE]{Do}{doWhile}{\algorithmicdo}[1]{\algorithmicwhile\ #1}%
\algnewcommand{\LineComment}[1]{\State \(\triangleright\) #1}
\newcommand*\diamonded[1]{\tikz[baseline=(char.base)]{
            \node[shape=diamond,draw,inner sep=1pt] (char) {$#1$};}}
\newcommand*\diam{\tikz[baseline=(char.base)]{
    \node[shape=diamond,draw,inner sep=1pt] (char) {\phantom{$i$}};}}

\makeatletter
\renewcommand*\env@matrix[1][*\c@MaxMatrixCols c]{%
  \hskip -\arraycolsep
  \let\@ifnextchar\new@ifnextchar
  \array{#1}}
\makeatother

\begin{document}
\title{Efficient Code for Relativistic Quantum Summoning}
\author{Ya-Dong Wu}
\affiliation{Institute for Quantum Science and Technology, University of Calgary, Alberta T2N 1N4, Canada}
\author{Abdullah Khalid}
\affiliation{Institute for Quantum Science and Technology, University of Calgary, Alberta T2N 1N4, Canada}
\author{Barry C.\ Sanders}
\affiliation{Institute for Quantum Science and Technology, University of Calgary, Alberta T2N 1N4, Canada}
\affiliation{Institute for Quantum Information and Matter, California Institute of Technology, Pasadena, California 91125, USA}
\affiliation{Program in Quantum Information Science, Canadian Institute for Advanced Research,Toronto, Ontario M5G 1Z8, Canada}
\affiliation{Shanghai Branch, National Laboratory for Physical Sciences at Microscale, University of Science and Technology of China, Shanghai 201315, China}
\begin{abstract}
Summoning retrieves quantum information,
prepared somewhere in spacetime,
at another specified point in spacetime,
but this task is limited by the quantum no-cloning principle and the speed-of-light bound.
We develop a thorough mathematical framework for summoning quantum information in a relativistic system and
formulate a quantum summoning protocol for any valid configuration of causal diamonds in spacetime.
For single-qubit summoning,
we present a protocol based on a Calderbank-Shor-Steane code
that decreases the space complexity for encoding by a factor of two compared to the previous best result and reduces the gate complexity from scaling as the cube to the square of the number of causal diamonds.
Our protocol includes decoding whose gate complexity scales linearly with the number of causal diamonds.
Our thorough framework for quantum summoning enables full specification of the protocol,
including spatial and temporal implementation and costs,
which enables quantum summoning to be a well posed protocol for relativistic quantum communication purposes.
\end{abstract}

\maketitle

\section{Introduction}

Quantum summoning is the task of encoding and transmitting quantum information to a configuration of spacetime causal diamonds such that the quantum information can be reconstructed in any one of these causal diamonds~\cite{Kent2013,JPA49,NJP18,PhysRevA.93.062327}. Quantum summoning cannot be guaranteed to work for every configuration of causal diamonds because quantum information cannot be copied~\cite{Park1970,wootters1982single,dieks1982,Juan} or transmitted superluminally~\cite{wootters1982single}. Summoning is only possible for a configuration of causal diamonds if every pair of diamonds is causally related, where two diamonds are causally related if the earliest point of one can communicate with the latest point of the other~\cite{JPA49}.
Our aim is to construct efficient protocols for summoning quantum information in any configuration of~$N$ pairwise-related causal diamonds.

A variety of work has been done on summoning ever since Kent introduced this task and presented a no-summoning theorem~\cite{Kent2013}. Hayden and May~\cite{JPA49} showed that quantum summoning can be reduced to the primitives of quantum secret sharing~\cite{PhysRevLett.83.648,PhysRevA.61.042311,PhysRevA.78.042309} and teleportation~\cite{PhysRevLett.70.1895,0305-4470-34-35-332}. They exploited a codeword-stabilized (CWS) quantum code~\cite{4729763} to design a summoning protocol that is efficient in the sense that the number of qubits $Q$ used by the code is polynomial in $N$. Hayden et al.~\cite{NJP18} proposed a continuous-variable version of summoning and an efficient protocol to perform this task, as well as showing that optical circuits can be used to realize summoning experimentally. In 2016, Adlam and Kent proposed a summoning task with multiple summonses and provided a protocol, which employs teleportation, to accomplish this task for the configuration being an ordered set of causal diamonds~\cite{PhysRevA.93.062327}.

Our protocol for summoning quantum information uses a Calderbank-Shor-Steane (CSS) code~\cite{PhysRevA.54.1098,Steane2551}
\begin{equation}
\label{eq:NN-12}
	\left[\left[
		{\tilde{N}\choose2},
		1,\frac{\tilde{N}}{2}\right]\right],\;
	\tilde{N}=2\left\lceil\frac{N}{2}\right\rceil\in2\mathbb{Z},
\end{equation}
that encodes one qubit into $\tilde{N}\choose2$ physical qubits, with the restriction that $\tilde{N}$ is even. 
We calculate that this CSS code distance~$\frac{\tilde{N}}{2}$,
and this code is constructed from the relation between graphs and linear algebra~\cite{diestel2005graph,doi:10.1063/1.2731356}.
Our code is a qubit version of the homological continuous-variable quantum error correcting code~\cite{NJP18}
and corrects erasure errors that occur in summoning. 

We provide a procedure to construct the encoding and  decoding circuits for our CSS code for any even positive integer~$\tilde N$.
The number of qubits~$Q$ used by our protocol is reduced by a factor of two compared to the previous best~\cite{JPA49}, 
and the number of quantum gates~$G$ is reduced from $O\left(N^3\right)$~\cite{JPA49} to~$O\left(N^2\right)$. 
Our decoding procedure has gate complexity~$O(N)$. 
Our results are significant in that we complete the quantum summoning protocol~\cite{Kent2013,JPA49,NJP18}
by providing both encoding and decoding schemes,
explain how to utilize quantum error correction for summoning, 
analyze quantum resources,
and demonstrate improved efficiency for our protocol.

  
Our paper is organized as follows. Section~\ref{sec:background} reviews the background knowledge regarding summoning, quantum error correction and algebraic graph theory. Section~\ref{sec:definition} provides the mathematical definition of summoning. In section~\ref{sec:protocol}, we study a protocol for summoning using a CSS code, including the encoding and the decoding methods, and the resource analysis of the CSS code and the CWS code. Sections~\ref{sec:discussion} and~\ref{sec:conclusion} give the discussion and the conclusion, respectively.

\section{Background}
\label{sec:background}

Here we explain the quantum information processing task of summoning and the conditions for the configurations to make summoning feasible~\cite{Kent2013,JPA49,NJP18}. 
Then we give the background knowledge on stabilizer quantum error correcting codes~\cite{gottesman1997stabilizer}, especially CSS codes~\cite{PhysRevA.54.1098,Steane2551}.
We introduce the relation between graphs and binary vectors~\cite{diestel2005graph}, which is useful for the construction of the CSS stabilizer code.
Finally, we briefly review the CWS code~\cite{4729763} used by Hayden and May to summon quantum information~\cite{JPA49}.

\subsection{Summoning}
\label{subsec:summoning}
Summoning is an information processing task involving Alice and Bob~\cite{Kent2013,JPA49,NJP18}. Bob's role is to provide the quantum information to Alice
and to designate where the quantum information is to be summoned and Alice's role is to summon
quantum information at the designated spacetime location.
Associated to each request point~$y$ is a reveal point~$z_y$ that is in the causal future of $y$.
The intersection of the future light cone of~$y$ with the past light cone of~$z_y$ is called a causal diamond, expressed as~$\diam$.
We label causal diamonds and show a label~$i$ in the diamond as $\diamonded{i}$.
Besides the request and reveal points, Alice and Bob also agree upon a starting point $s$, where Bob provides the quantum information to Alice.

Alice and Bob can arrange their agents at various points in spacetime prior to the start of summoning~\cite{NJP18}.
 Bob designates one agent to be the referee who sends quantum information to point $s$ and classical information to all the request points. 
 Alice designates one agent to be the starting agent~$S$, who is situated at point $s$, and she delegates agents to each request and reveal point. 
 We label the agent at point~$x$ by $A_x$.
 Figure~\ref{fig:causaldiamond} shows an example of Alice's and Bob's agents arranged in spacetime.

\begin{figure}
\begin{center}
\includegraphics[width=1\linewidth]{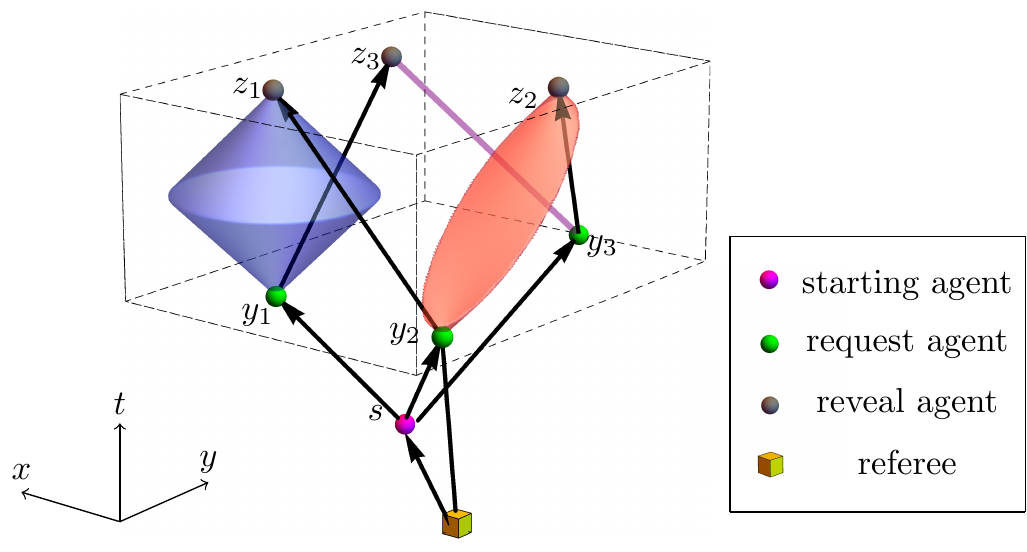}
\end{center} 
\caption{%
Three causal diamonds (red, blue and purple) in spacetime.
A referee, a starting agent, three request agents and three reveal agents are arranged in spacetime. 
An arrow represents a quantum communication channel from one agent to another agent, and a line segment between two agents represents a classical channel from one to the other. 
The referee sends a quantum state~$\ket{\psi}$ to the starting agent, and randomly chooses $y_2$ to send a classical request to~$A_{y_2}$. 
The starting agent encodes~$\ket{\psi}$ to three qutrits and distribute them to three request agents respectively. 
$A_{y_2}$ sends her qutrit to~$A_{z_2}$. Receiving no request,
 the request agents at~$y_1$ and~$y_3$ send their qutrits to~$A_{z_3}$ and~$A_{z_2}$ respectively. 
 Hence,~$A_{z_2}$ receives two qutrits and decodes the state~$\ket{\psi}$.}
\label{fig:causaldiamond}
\end{figure}

When summoning starts, the referee prepares a quantum state
\begin{equation}
\label{eq:ketpsi}
	\ket{\psi}\in \mathscr{H},
\end{equation}
where~$\mathscr{H}$ is a finite $d$-dimensional Hilbert space~\cite{JPA49},
and transmits~$\ket{\psi}$ to the starting agent. 
Alice and all her agents do not have any knowledge of~$\ket{\psi}$.  
The referee randomly chooses one request point, say~$y$, and sends the request only to~$A_y$. 
Then Alice's task is to present the quantum state~$\ket{\psi}$ at the corresponding reveal point $z_y$, by her agents' collaboration.

Given a set of causal diamonds $\left\{\diamonded{i} \right\}_{i=1}^{N}$, summoning might be infeasible~\cite{Kent2013} due to the restrictions of both the no-cloning theorem~\cite{Park1970,wootters1982single,dieks1982,Juan} and no superluminal communication~\cite{wootters1982single}. 
Summoning is possible if and only if the following two conditions are satisfied~\cite{JPA49}.  
\begin{itemize}
\item[C1] All reveal points are in the causal future of~$s$. 
\item[C2] Each pair of causal diamonds is causally related, which means that there exists a point in one causal diamond
that is causally related with at least one point in the other causal diamond. 
\end{itemize}

We call a set of causal diamonds satisfying these two conditions a ``valid configuration" for summoning. 
We represent a configuration of causal diamonds by a graph~$G$ as follows~\cite{JPA49,NJP18}. 
We assign each causal diamond to a vertex and use the label of the causal diamond to label the vertex. 
If two causal diamonds are causally related, an edge~$e$ is inserted between the two corresponding vertices in~$G$. 
A valid configuration of~$N$ causal diamonds is represented by an~$N$-vertex complete graph
denoted~$K_N$,
for which each pair of vertices is connected by an edge~\cite{diestel2005graph}.
 
In Fig.~\ref{fig:causaldiamond}, we present an example of using quantum secret sharing~\cite{PhysRevLett.83.648,PhysRevA.61.042311,PhysRevA.78.042309} 
to summon quantum information~\cite{JPA49}. After receiving~a qubit $\ket{\psi}$, the starting agent encodes~$\ket{\psi}$ into three qutrits~\cite{PhysRevLett.83.648} and 
 distributes the three qutrits to the three request agents.
 If~$A_{y_i}$ ($i=1,2,$ or $3$) receives the request, then the request agent sends her qutrit to the reveal point
\begin{equation}
\label{eq:revealpoint}
	z_i\coloneqq z_{y_i}.
\end{equation}
Otherwise, she sends her qutrit to the reveal point $z_{(i-1)\!\!\!\mod\!3}$. In such a way, no matter which request agent receives the request, 
 the associated reveal agent receives two qutrits to retrieve the original qubit~$\ket{\psi}$.

\subsection{Stabilizer codes}
\label{subsec:stabilizer}

In this subsection, we begin by introducing the Pauli group and the representation of Pauli operators using binary vectors. 
We explain the parameters characterizing a quantum error correcting code and the definition of erasure errors.
Then we explain stabilizer codes~\cite{gottesman1997stabilizer}, specially CSS codes~\cite{PhysRevA.54.1098,Steane2551}. Finally, we show the encoding of a stabilizer code.

An $n$-qubit Pauli group~\cite{preskill1998lecture} is
\begin{equation}
\mathcal G_n \coloneqq  \pm \{I, X, Y, Z\}^{\otimes n},
\end{equation}
where 
\begin{equation}
X\coloneqq \begin{bmatrix} 0 & 1 \\ 1 & 0 \end{bmatrix}, \,Z\coloneqq \begin{bmatrix} 1 & 0 \\ 0 & -1\end{bmatrix}, \, \text{and } Y\coloneqq ZX.
\end{equation}
The Pauli group module sign is isomorphic to a cartesian product of binary vector spaces according to 
\begin{equation}
\label{symplecticnotation}
\mathcal{I}_P: \mathcal{G}_n/\mathbb{Z}_2 \rightarrow \mathbb{Z}_2^n \times \mathbb{Z}_2^n: \bigotimes_{i=1}^n Z^{u_i} X^{v_i} \mapsto \begin{bmatrix} \bm{u} & \bm{v} \end{bmatrix},
\end{equation}
where 
\begin{equation}
\bm{u}\coloneqq \left[u_1 \dots u_n \right], \;
\bm{v}\coloneqq \left[v_1 \dots v_n \right].
\end{equation}  
 Two Pauli operators represented as binary vectors,
$\begin{bmatrix} \bm{u} & \bm{v} \end{bmatrix}$ and
$\begin{bmatrix} \bm{u}' & \bm{v}' \end{bmatrix}$,
mutually commute if and only if~\cite{preskill1998lecture}
\begin{equation}
\label{commutecondition}
	\bm{u} \cdot \bm{v}' + \bm{v} \cdot \bm{u}' =0,
\end{equation}
where~$\cdot$ is the indefinite inner product
\begin{equation}
\label{innerproduct}
\bm{u}\cdot \bm{v}\coloneqq \sum_{i=1}^n u_i v_i\in \mathbb{Z}_2.
\end{equation}
Otherwise, these two Pauli operators anti-commute with each other.

In quantum error correcting codes, $[[n, k, d]]$ denotes a quantum code, 
where~$k$ qubits are encoded into~$n$ qubits, and~$d$ is the distance of the code~\cite{preskill1998lecture}.
Given a Pauli operator $P\in \mathcal G_n$, the weight of~$P$ is the number of nonidentity single-qubit Pauli operators, i.e., $X$,~$Y$ and~$Z$ in the tensor product~$P$. 
The distance of a quantum error correcting code is the minimum weight of a Pauli operator $P$ such that
\begin{equation}
\label{distance}
\langle i|P\ket{j}\neq C(P)\delta_{ij},
\end{equation}
where $\ket{i}$ and $\ket{j}$ are basis elements of the code, $C(P)$ is a constant depending on~$P$, and $\delta_{ij}$ is the Kronecker delta function.
When transmitting a block of $n$ qubits, if $t$ qubits ($t<n$) are lost or never received, while the other $n-t$ qubits are undamaged, then the errors at these $t$ qubits are called erasure errors~\cite{PhysRevA.56.33}.

A stabilizer code~\cite{gottesman1997stabilizer} is the simultaneous eigenspace of all the elements of an Abelian subgroup $\mathcal S$ of $\mathcal G_n$ with eigenvalue one.
A generator set of~$\mathcal S$ is a set of independent elements in~$\mathcal S$ 
such that every element of~$\mathcal S$ can be expressed as a product of the elements in this generator set.
An $[[n, k, d]]$ stabilizer code has $n-k$ independent stabilizer generators, each of which can be represented by a $2n$-dimensional binary vector. 
The $[[n, k, d]]$ stabilizer code is characterized by an $(n-k)\times 2n$ stabilizer generator matrix, where each row represents a stabilizer generator.

Now we introduce CSS codes, which are a type of stabilizer codes. 
A CSS code~\cite{PhysRevA.54.1098,Steane2551} is specified by two classical linear codes $\mathcal{C}_1$ and $\mathcal{C}_2$, where $\mathcal{C}_2$ is a subcode of $\mathcal{C}_1$, i.e.,
\begin{equation}
\mathcal{C}_2 \subseteq \mathcal{C}_1.
\end{equation}
Each basis element of a CSS code corresponds to a coset~\cite{roman2005advanced} of~$\mathcal{C}_2$ in~$\mathcal{C}_1$, where the basis element is an equally weighted superposition of all the codewords in the coset.
A CSS code is a stabilizer code whose stabilizer generators are either tensor products of $X$ operators and identities, or tensor products of $Z$ operators and identities~\cite{gottesman1997stabilizer}. 
Hence, the CSS stabilizer code is characterized by an $(n-k)\times 2n$ stabilizer generator matrix 
\begin{equation}
\label{CSScode}
	\begin{bmatrix}
		H_Z &\bm0 \\
		\bm0 &H_{X}
	\end{bmatrix},
\end{equation}
where $H_Z$ and $H_X$ are two matrices, and the $\bm0$s are appropriately sized zero matrices.

Here we explain the encoding of a stabilizer code with stabilizer~$\mathcal S$.
The Pauli operators that preserve the stabilizer code space but act nontrivially on the encoded state are the logical Pauli operators on the encoded state~\cite{gottesman1997stabilizer}. 
The logical Pauli operators commute with all stabilizers in~$\mathcal S$ but lie outside~$\mathcal S$. 

For an $[[n ,1, d]]$ stabilizer code with stabilizer~$\mathcal S$, we denote $\bar{Z}$ and $\bar{X}$ as the logical $Z$ and logical $X$ operators on an $n$-qubit encoded state. 
Suppose
\begin{equation}
\bar{Z} \ket{\psi_0}=\ket{\psi_0},
\end{equation}
and $$\left\{M_i\right\}_{i=1}^{n-1}$$ are $n-1$ independent stabilizer generators of~$\mathcal S$.
The encoded logical states are~\cite{gottesman1997stabilizer}
\begin{equation}
\label{encodingzero}
\ket{0}_L\coloneqq \prod_{i=1}^{n-1} (I+M_i) \ket{\psi_0},
\end{equation}
and  
\begin{equation}
\label{encodingone}
\ket{1}_L\coloneqq \bar{X}\prod_{i=1}^{n-1} (I+M_i) \ket{\psi_0}=\prod_{i=1}^{n-1} (I+M_i) \bar{X}\ket{\psi_0}.
\end{equation}
Equations (\ref{encodingzero}) and~(\ref{encodingone}) indicate how to encode one qubit by a stabilizer code.

We have discussed quantum error correction, especially stabilizer codes in this subsection. 
Next we study the close relation between a graph and a binary vector, which provides a powerful tool to construct the CSS stabilizer code.

\subsection{Graphs and linear algebra}

In this subsection, 
we begin by defining graphs and the binary linear space.
After explaining these two concepts, 
 we describe an isomorphism from sets of edges of
 an $n$-vertex graph to binary vectors with length $n\choose2$~\cite{diestel2005graph}.
Our approach is inspired by homology theory to construct quantum error correcting codes~\cite{NJP18,doi:10.1063/1.2731356}.
Finally, we present the examples of triangle graphs and star graphs.

A graph~\cite{diestel2005graph}
\begin{equation}
G\coloneqq \left( V, E\right)
\end{equation}
comprises a set of vertices~$V$ and a set of edges
\begin{equation}
E\subseteq V \times V.
\end{equation}
One example of a graph is the $n$-vertex complete graph,~$K_{n}$.

To explain the binary linear space, we introduce $GF(2)$, which is the smallest finite field containing two elements $\{0, 1\}$, together with addition and multiplication operations~\cite{lidl1997finite}.
The linear space~\cite{roman2005advanced} over field $GF(2)$, denoted by $\mathbb{Z}_2^m$, is a set $\{0,1\}^m$, together with vector addition,
\begin{equation}
+: \mathbb{Z}_2^m\times \mathbb{Z}_2^m \rightarrow \mathbb{Z}_2^m,
\end{equation}
 and scalar multiplication,\footnote{Note we use $\cdot$ for the scalar multiplication only in Eq.~(\ref{scalarmultiplication}) and Table~\ref{isomorphism}. After this subsection, we use $\cdot$ only for the indefinite inner product.}
\begin{equation}
\label{scalarmultiplication}
\cdot : GF(2)\times \mathbb{Z}_2^m \rightarrow\mathbb{Z}_2^m.
\end{equation}
 Next we explain the relation between an edge set~$E$ of a graph~$G$ and a binary vector with length~$|V|\choose2$, where $|V|$ is the cardinality of~$V$.

Given $K_n=(V_K, E_K)$, the power set of~$E_K$, which is the set of all the subsets of $E_K$, forms a binary linear space~$\mathscr{E}$~\cite{diestel2005graph}.  
The power set of $E_K$ is denoted $2^{E_K}$.
For $U$, $U' \in 2^{E_K}$, the addition of~$U$ and~$U'$ amounts to the symmetric difference of~$U$ and~$U'$,
\begin{equation}
U+U'\coloneqq  \left( U \cup U' \right)\setminus \left(U \cap U' \right).
\end{equation}
The empty set~$\O$ is the zero element and 
\begin{equation}
\forall\, U\in 2^{E_K},\, -U\coloneqq U.
\end{equation}

For $\diamonded{i}$ and $\diamonded{j}$,
an edge $e_{ij}$ is a unit vector in~$\mathscr{E}$, and $e_{ji}=e_{ij}$ because we are dealing with undirected graphs. 
The set of edges
\begin{equation}
\label{eq:edges}
	\left\{e_{ij}\right\}_{1\le i<j \le n}
\end{equation}  
forms an orthonormal basis of~$\mathscr E$. 
As ${n\choose2}$ edges exist in~$K_{n}$, 
\begin{equation}
	\operatorname{dim}\mathscr{E}
		={n\choose2}.
\end{equation}

Now we show that~$\mathscr{E}$ is isomorphic to~$\mathbb{Z}_2^{n\choose2}$~\cite{diestel2005graph}.
Given any $U\in \mathscr E$, an isomorphism is
\begin{equation}
\label{iso:graphvector}
\mathcal{I}_G: \mathscr E\rightarrow  \mathbb{Z}_2^{n\choose2} : U\mapsto \bm{u}=\left[ u_1\, \dots\, u_{n\choose2} \right],
\end{equation}
where $$\{u_i\}_{i=1}^{n\choose2}$$ are the coefficients of~$U$ with respect to the basis in Eq.~(\ref{eq:edges}).
The isomorphic mappings of the vector addition and the scalar multiplication are shown in Table~\ref{isomorphism}.

\begin{table}[b]
\begin{center}
\begin{tabular}{|c|c|c|}
\hline 
 & $\mathscr{E}$ & $\mathbb{Z}_2^{n\choose 2}$ \\
 \hline
 $+$ & $\left( U \cup U' \right)\setminus \left(U \cap U' \right) $ &$ \bm{u}+\bm{u}'$ \\
 \hline
 $\cdot$ & $0 U=\O, \, 1 U=U $ & $0 \bm{u}=\bm{0},\, 1 \bm{u}=\bm{u}$ \\
\hline 
\end{tabular}
\end{center}
\caption{The mapping of the vector addition and the scalar multiplication on $\mathscr E$ to those operations on $\mathbb{Z}_2^{n\choose 2}$}
\label{isomorphism}
\end{table}

\begin{figure}
\begin{center}
\includegraphics[width=1\linewidth]{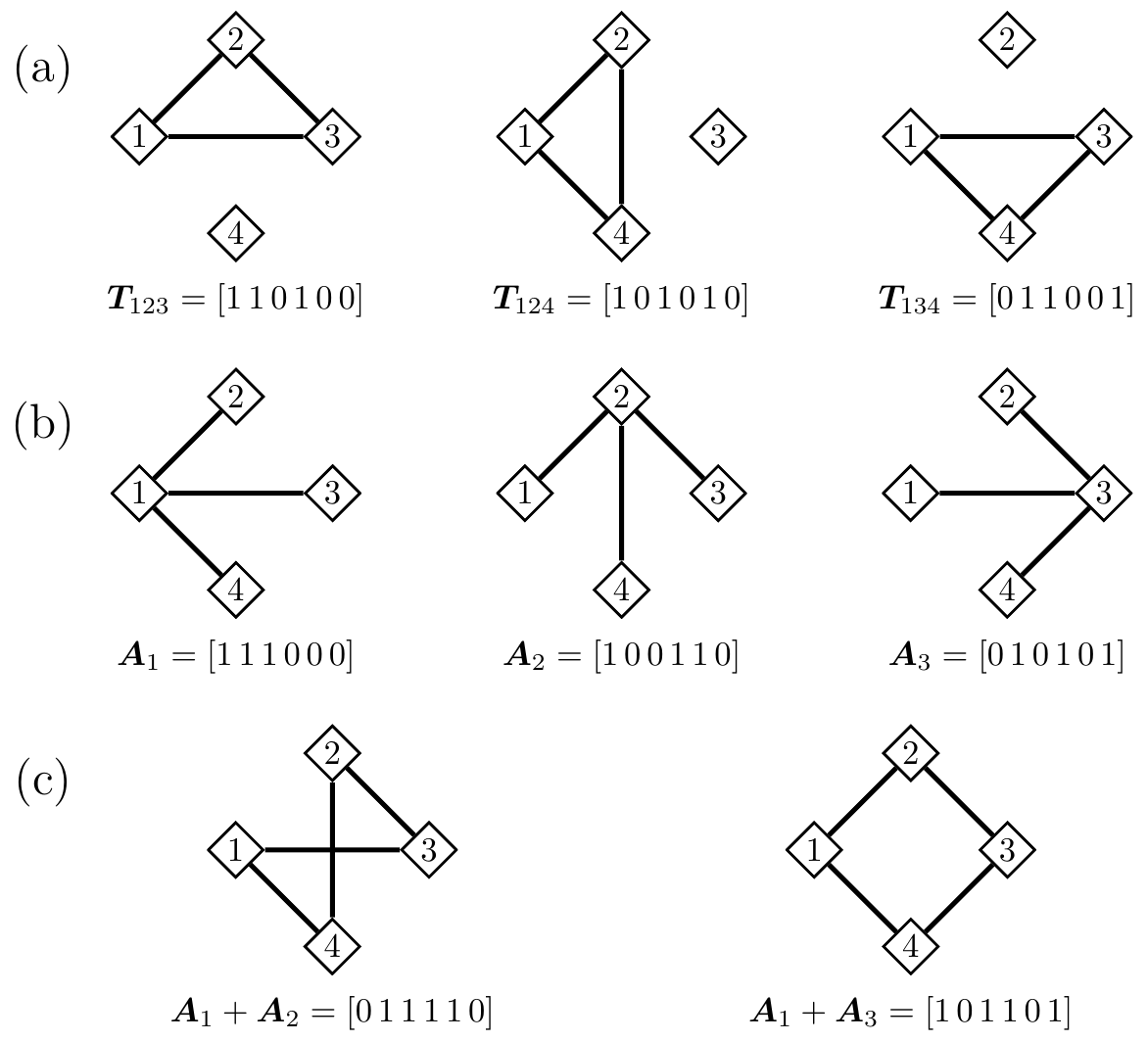}
\end{center} 
\caption{For $n=4$, (a) the triangle graphs representing $\bm{T}_{1jk}$ ($2\le j < k \le 4$), (b) the star graphs representing $\bm{A}_l$ ($1\le l\le 3$) and (c)
the graphs representing $\bm{A}_1+\bm{A}_m$ ($2\le m\le 3$).}
\label{fig:CSSgraph}
\end{figure}

Here we introduce two types of $n\choose2$-dimensional binary vectors and two linear subspaces spanned by these two types of vectors as examples of the isomorphism~(\ref{iso:graphvector}). 
These examples are used later for the construction of the CSS stabilizer code.
The two types of vectors in~$\mathbb{Z}_2^{n\choose2}$ are 
\begin{equation}
 \bm{T}_{ijk} \coloneqq \bm{e}_{ij}+\bm{e}_{jk}+\bm{e}_{ki},\;
  \bm{A}_{j} \coloneqq \sum_{1\le l \le n, \,  l\not=j }\bm{e}_{lj},
\end{equation}
where $\bm{e}_{ij}$ is the unit vector $\mathcal {I}_{G}(e_{ij})$.
From the isomorphism~(\ref{iso:graphvector}), these two types of vectors can be represented by two different types of graphs.  
~$\bm{T}_{ijk}$ is represented by a triangle graph connecting vertices
\begin{equation}
	\left\{\diamonded{i},\diamonded{j},\diamonded{k}\right\},
\end{equation}
and~$\bm{A}_{j}$ is represented by a star graph with vertex~$\diamonded{j}$ connected to every other vertex.

We construct the linear space 
\begin{equation}
\label{eq:cutspace}
	\mathcal{C}_1
		\coloneqq  \operatorname{span} \left\{\bm{A}_1, \bm{A}_2, \dots, \bm{A}_{n-1}\right\}
\end{equation}
spanned by $n-1$ linearly independent $\left\{\bm{A}_{j}\right\}$,
and the orthogonal linear space
\begin{equation}
\label{eq:cyclespace}
	\mathcal{C}_1^\perp
		\coloneqq \operatorname{span}
			\left\{\bm{T}_{123},\bm{T}_{124}, \dots,\bm{T}_{12n},\bm{T}_{134},
				\dots,\bm{T}_{1\, n-1\, n}\right\}
\end{equation}
spanned by ${n-1\choose2}$ linearly independent $\left\{\bm{T}_{ijk}\right\}$.
The elements in~$\mathcal C_1^\perp$ are represented by Eulerian cycles (graph cycles that use each edge exactly once)~\cite{diestel2005graph}.
Meanwhile, $\mathcal{C}_1$ comprises vectors orthogonal to all vectors in~$\mathcal{C}_1^\perp$~\cite{diestel2005graph}, i.e.,
\begin{equation}
\label{perpofperp}
\mathcal{C}_1=\left(\mathcal{C}_1^{\perp}\right)^\perp.
\end{equation}

We introduce an ($n-2$)-dimensional linear subspace of~$\mathcal{C}_1$
\begin{equation}
\label{eq:pairs}
	\mathcal{C}_2
		\coloneqq \operatorname{span}
			\left\{\bm{A}_1+ \bm{A}_2, \bm{A}_1+\bm{A}_3 \dots,\bm{A}_1+ \bm{A}_{n-1}\right\}
				\subset\mathcal{C}_1.
\end{equation}
$\mathcal C_2$ together with $\mathcal C_1$ specifies the CSS code for summoning in Subsec.~\ref{subsec:stabilizer}.
Figure~\ref{fig:CSSgraph}(a), (b) and (c) depict the graphs representing bases of linear spaces~$\mathcal{C}_1^\perp$,~$\mathcal C_1$ and~$\mathcal C_2$ respectively for $n=4$.

This subsection has shown that the power set of the edge set of an $n$-vertex complete graph forms a binary linear space, isomorphic to $\mathbb{Z}_2^{n\choose2}$. 
Hence, we have constructed a graph representation of any $n\choose2$-binary vector. 
The examples given in this subsection are useful to construct the CSS code for summoning.

\subsection{The CWS code for summoning}
 In this subsection, we first introduce CWS codes~\cite{4729763}, then explain the graph-state formalism of CWS codes. Finally we discuss the CWS code used by Hayden and May for quantum summoning~\cite{JPA49}.  In Subsec.~\ref{subsec:comparison}, we study the gate complexity of the encoding of the CWS code for summoning and compare it with the CSS code.
 
An $((n,k))$ CWS code~\cite{4729763} encodes a $k$-dimensional Hilbert space to $n$ qubits. This CWS code is specified by a word stabilizer, which is a $2^n$-element Abelian subgroup~$\mathcal S$ of $\mathcal G_n$, 
and a set of $k$ word operators, which are $n$-qubit Pauli operators $$\{W_l\}_{l=1}^k.$$  
The word stabilizer $\mathcal S$ specifies a unique $\ket{\psi_{\mathcal S}}$ such that $\forall M\in\mathcal S$,
\begin{equation}
M\ket{\psi_{\mathcal S}}=\ket{\psi_{\mathcal S}}.
\end{equation}
The CWS code is spanned by the basis 
\begin{equation}
\{\ket{w_l}\coloneqq  W_l \ket{\psi_{\mathcal S}}\}_{l=1}^k.
\end{equation}

Under local Clifford operations, any CWS code is equivalent to its standard form~\cite{4729763}, whose word stabilizer is a graph-state stabilizer~\cite{PhysRevA.68.022312}, 
and whose word operators contain only $Z$ operators and identities.
Thus, the stabilized state~$\ket{\psi_{\mathcal S}}$ of a CWS code in its standard form is a graph state.
Given a graph~$G=(V,E)$, the associated graph state is~\cite{PhysRevA.68.022312}
\begin{equation}
\ket{G}=\prod_{(i,\, j)\in E}CZ_{(i,\,j)}  H^{\otimes |V|} \ket{0}^{|V|},
\label{graphstate}
\end{equation}
 where~$CZ_{(i\,j)}$ is the controlled-$Z$ gate with control qubit~$i$ and target qubit~$j$, and~$H$ is the Hadamard gate.
 
  \begin{figure}
\includegraphics[width=0.8\linewidth]{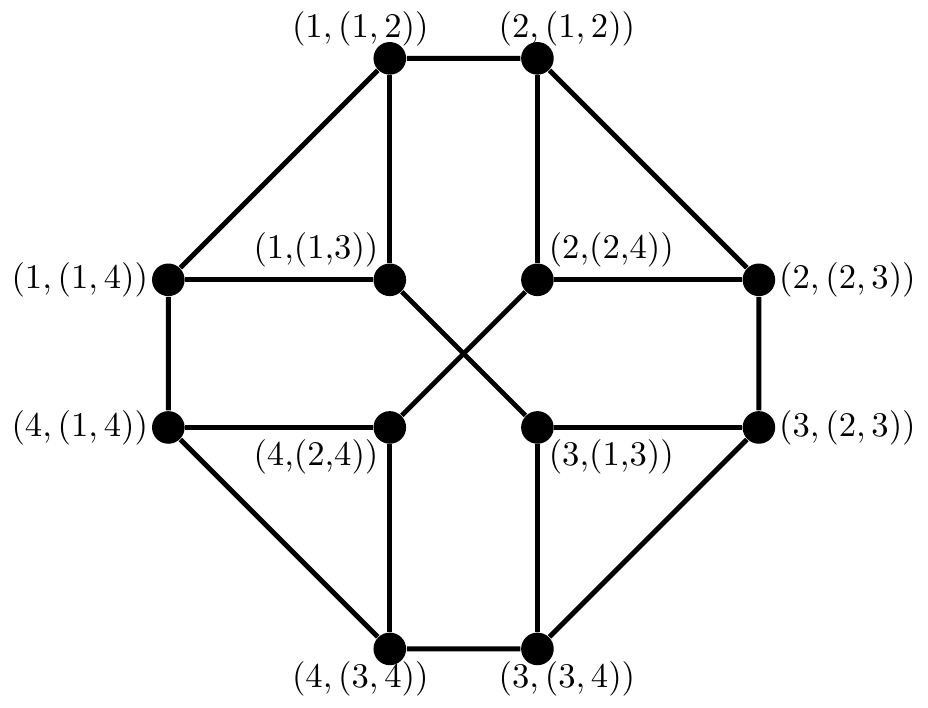}
 \caption{$G_{\text{CWS}}$ for $N=4$. Each vertex of~$G_{\text{CWS}}$ is labeled by $(j, (j, k))$ for $1\le j, k\le 4$ and $k\neq j$. Each $(j, (j, k))$ is adjacent to $(k,(j, k))$ and $(j,(j,l))$, where $1\le l \le 4$ and $l\neq j$ or $k$.}
\label{fig:CWSgraph}
 \end{figure}

Hayden and May propose a $\left(\left(2{N\choose2},2\right)\right)$ CWS code to summon a qubit in~$N$ causal diamonds~\cite{JPA49}.
This $\left(\left(2{N\choose2},2\right)\right)$ CWS code is specified by a graph-state stabilizer represented by a graph $G_{\text{CWS}}$ 
and two word operators 
\begin{equation}
\label{codewordoperators}
\left\{I, Z^{\otimes N(N-1)}\right\}.
\end{equation}
Given an $N$-vertex complete graph $K_{N}=\{V_K, E_K\}$, $G_{\text{CWS}}$ is the line graph~\cite{diestel2005graph} of $G'\coloneqq \{V', E'\}$, where 
\begin{equation}
V'=V_K\cup E_K,
\end{equation}
and
\begin{equation}
E'=\left\{(v,(v,w)); v\in V_K, (v,w)\in E_K\right\}.
\end{equation}
Figure~\ref{fig:CWSgraph} presents~$G_{\text{CWS}}$ for~$N=4$.
 This CWS code can be used to summon one qubit in four causal diamonds~\cite{JPA49} by employing twelve qubits. 

This section has introduced quantum summoning as a quantum information processing task in spacetime and explained the conditions on the configurations to make quantum summoning feasible.
We have briefly reviewed the background knowledge on quantum error correction, especially stabilizer codes and CSS codes, which are used for quantum summoning in this paper.
We have explained the connection between graphs and binary vectors, which is useful for the construction of the CSS code for summoning.
Finally, we have reviewed the CWS code used for quantum summoning~\cite{JPA49}, with which we compare our CSS code in Subsec.~\ref{subsec:comparison}.

\section{Mathematical Definition of Summoning}
\label{sec:definition}

In this section, we mathematically define both classical and quantum summoning. 
We begin by formalizing the notions of past and future light cones and causal diamonds. 
We than establish notations for a configuration of causal diamonds and the sets of request and reveal points.
Subsequently, we give a careful definition of both classical and quantum summoning, and when these tasks are trivial. 
 
Each spacetime point is $x\in \mathcal M$,
where~$\mathcal M$ denotes Minkowski spacetime~\cite{spacetime}. 
The future light cone for $x$ is 
\begin{equation}
        \operatorname{fut}(x)\coloneqq \left\{ w\in \mathcal M; w\succ x\right\},
\end{equation}
where $w \succ x$ indicates that information can be sent from~$x$ to $w$.
 The past light cone for $x$ is
\begin{equation}
	\operatorname{pas}(x)\coloneqq \left\{w\in \mathcal M; w\prec x \right\},
\end{equation}
where $w\prec x$ indicates that information can be received at $x$ from $w$.
A causal diamond for a pair of points $(y_i, z_i)\in \mathcal M \times \mathcal M$ satisfying $y_i\prec z_i$ is 
\begin{equation}
\diamonded{i} \coloneqq \left\{x\in \mathcal M; x\in\operatorname{fut}(y_i)\cap \operatorname{pas}(z_i)\right\}.
\end{equation}
 A configuration of causal diamonds is
\begin{equation}
\label{eq:configuration}
	 \mathscr{C}\coloneqq  \left\{\diamonded{i}; y_i\prec z_i\right\},
\end{equation}
  and $N\coloneqq |\mathscr C|$. 
 Two causal diamonds~$\diamonded{i}$ and~$\diamonded{j}$ are causally related if and only if $\exists \,x\in \diamonded{i}$, $\exists \, w\in \diamonded{j}$ such that either $x\in \operatorname{fut}(w)$, or $x\in\operatorname{pas}(w)$.
  The set of request points is 
\begin{equation}
 \text{REQ}\coloneqq \left\{y_i\in \mathcal M;\diamonded{i}\in \mathscr C\right\},
\end{equation} 
 and the set of reveal points is 
\begin{equation}
 \text{REV}\coloneqq \left\{z_i\in \mathcal M; \diamonded{i}\in \mathscr C\right\}.
\end{equation} 
 A starting point is $s\in\mathcal M$ such that $\forall \, z\in \text{REV}, z\in \operatorname{fut}(s)$, where  
there is a starting agent $S$.

We now formalize Kent's classical summoning protocol~\cite{Kent2013} by making each object mathematically well defined.
Given starting agent~$S$ at $s\in\mathcal{M}$,
request agents 
\begin{equation}
	\text{REQAG}\coloneqq \left\{A_y; y\in\text{REQ}\right\}
\end{equation}
and corresponding reveal agents 
\begin{equation}
	\text{REVAG}\coloneqq \left\{A_z; z\in\text{REV}\right\},
\end{equation}
and~$S$ possessing $n$-bit string $m\in\{0,1\}^n$,
summoning is the task of delivering~$m$ to any agent in REVAG
given arbitrary external selection of some~$y\in\text{REQ}$, which is
only revealed at spacetime point~$y$.

\begin{remark}
Classical summoning is trivial because $S$ broadcasts $m$ to all $z\in \text{REV}$~\cite{Kent2013}.
\end{remark}

The notion of quantum summoning~\cite{Kent2013,JPA49} builds on the concept of classical summoning,
which we formalize as follows.
Given a starting agent~$S$, REQAG and REVAG, and~$S$ possessing quantum information
\begin{equation}
	\ket{\Psi}\in \mathscr{H}_2^{\otimes n}
\end{equation}
(with $S$ possibly oblivious to~$\ket{\Psi}$), summoning is the task of delivering~$\ket{\Psi}$ to any agent in REVAG given arbitrary external selection of some~$y\in\text{REQ}$, which is
only revealed at spacetime point~$y$.

\begin{remark}
Summoning is trivial if S has a classical description of~$\ket{\Psi}$ because S broadcasts this description such that all agents in REVAG receive and can reconstruct~$\ket{\Psi}$.
\end{remark}
\begin{remark}
Quantum summoning is trivial if there is a causal curve, which starts from~$s$ and runs sequentially through all $\diamonded{i}\in\mathscr C$ in any order.
The protocol is trivial in this case
because quantum information can simply be sent along this causal curve.
When quantum information arrives at $\diamonded{j}$, $A_{y_j}$ decides whether to send it to $z_j$ or to send it to the next causal diamond depending on whether she receives  the request or not.
\end{remark}

\section{Summoning by Quantum Error Correction}
\label{sec:protocol}
Here we present a protocol for Alice to summon quantum information. 
We specify the actions that the starting agent and each of the request and reveal agents perform to fulfill any summoning request~(Subsec.~\ref{subsec:protocol}). 
For any valid configuration of causal diamonds, we propose a CSS code for the protocol of quantum summoning~(Subsec.~\ref{subsec:stabilizer}).
The encoding and decoding circuits of the CSS code are provided~(Subsec.~\ref{subsec:encodedecode}).  
 We show that the CSS code consumes fewer quantum resources than the CWS code~\cite{JPA49}~(Subsec.~\ref{subsec:comparison}).

\subsection{Protocol}
\label{subsec:protocol}

\begin{figure*}[t]
\begin{center}
\includegraphics[width=0.95\linewidth]{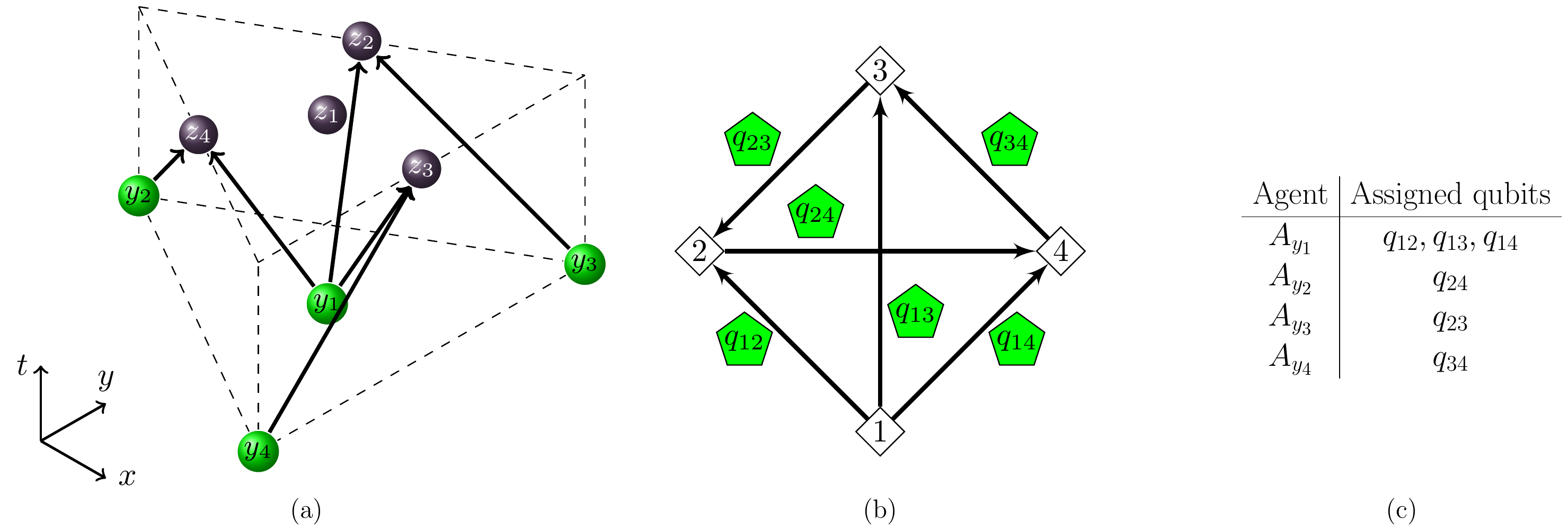}
\end{center}
\caption{(a) A configuration of four causal diamonds in $2+1$ dimensions.
Three request points ($y_2,y_3,y_4$) are placed at the base vertices of
a equilateral triangular prism and a fourth ($y_4$) is placed at the
centroid of base vertices. The reveal points are placed at the midpoints
of the top vertices ($z_1, z_2, z_3$) and the centroid of the top
vertices. The volume of the diamond is not shown for visual clarity. The
black arrows represent causal connections between points. (b) A complete graph representing the causal connections between
the diamonds depicted in (a). For the
CSS code the qubit $q_{ij}$ is assigned to edge $e_{ij}$. (c) A table showing which requests agents is each physical qubit sent to.}
\label{fig:fourvertextournamentprism2}
\end{figure*}

 Here we propose a protocol using the 
 CSS code~(\ref{eq:NN-12})
for summoning one qubit in any valid spacetime configuration.
This CSS code assigns one qubit to each edge of the complete graph~$K_{\tilde{N}}$;
hence, the number of qubits used by the protocol is
\begin{equation}
\label{eq:Q}
	Q= {\tilde{N}\choose2},
\end{equation}
for~$N$ and~$\tilde N$ related according to Eq.~(\ref{eq:NN-12}).

For a spacetime configuration with an even number~$N$ of causal diamonds, 
$\tilde{N}\coloneqq N$
and~$S$ employs the CSS code~(\ref{eq:NN-12})
to encode a qubit,
\begin{equation}
\label{qubitstate}
\ket{\psi}=\alpha \ket{0} +\beta \ket{1},
\end{equation}
 into~$Q$ (\ref{eq:Q})
  qubits and assigns each qubit to an edge of the complete graph~$K_{\tilde{N}}$.
The qubit assigned to the edge $e_{ij}$ is called~$q_{ij}$, where $e_{ij}$ denotes the edge connecting~$\diamonded{i}$ to~$\diamonded{j}$ for $i,j\in[N]$ and $i\neq j,$.

$S$ sends $q_{ij}$ to $A_{y_i}$ if 
\begin{equation}
y_i \prec z_j,
\end{equation}
 and to $A_{y_j}$ if 
\begin{equation}
  y_j \prec z_i.
\end{equation}
  If $A_{y_i}$ receives the summoning request,  she sends all the qubits in her possession to $A_{z_i}$. 
  Otherwise, she sends each qubit $q_{ij}$ in her possession to~$A_{z_j}$. 
  As each vertex is adjacent to $N-1$ edges, any reveal agent, who receives the summoning request, receives $N-1$ qubits. 
  Later, we prove that $\forall \, r\in [N]\coloneqq (1\,2\,\cdots N)$, the qubits  
  \begin{equation}
  \label{accesseven}
 \left \{q_{rk}; k\in[N]\setminus \{r\} \right \}
  \end{equation} can be used to decode~$\ket{\psi}$ perfectly.
  {Fig.~(\ref{fig:fourvertextournamentprism2}) shows how the qubits are assigned to the request agents for a configuration of four causal diamonds.

 In a configuration of an odd number of causal diamonds,
\begin{equation}
 \tilde N\coloneqq N+1.
\end{equation}
 $S$ introduces one more vertex $$\diamonded{N \!\!+\!\!1}$$ to obtain graph $K_{N+1}$. 
 This new vertex can be seen as ficticious causal diamond causally related with every causal diamond, but the summoning request is never sent to this causal diamond.
 Then~$S$ employs the CSS code~(\ref{eq:NN-12}),
which encodes~$\ket{\psi}$ into ${N+1\choose2}$ qubits. 
As before,~$S$ sends each qubit $q_{ij}$, where $i,j\in [N]$,
to $A_{y_i}$ if
\begin{equation}
y_i\prec z_j,
\end{equation}
  and to $A_{y_j}$ if 
\begin{equation}
 y_j \prec z_i.
\end{equation}
$S$ sends each additional qubit~$q_{j\, N+1}$ to reveal agent $A_{z_j}$.
As in the even case, if $A_{y_i}$ receives the summoning request,~$A_{y_i}$ sends all the qubits in her possession to $A_{z_i}$. 
  Otherwise, she sends each qubit~$q_{ij}$ in her possession to~$A_{z_j}$. 
 Any reveal agent who receives the summoning request,
 ultimately receives~$N$ qubits to decode~$\ket{\psi}$. 
 For any $r\in [N+1]$, the qubits  
  \begin{equation}
  \label{accessodd}
  \left\{q_{rk}; k\in[N+1]\setminus \{r\} \right\}
  \end{equation} can be used to decode~$\ket{\psi}$ perfectly. 
  
  This protocol can also be used to accomplish a variant of the summoning task~[4]. In this modified task, the set of causal diamonds are causally ordered and the request may be sent to multiple request agents but only one of them needs to comply with the summon. If the causal diamonds are causally ordered, then in the subset of request agents to whom the requests are sent,
\begin{equation}
  \{A_{y_i}; i\in K\subseteq [N] \},
\end{equation}
let $A_{y_r}$ be the earliest request agent according to the causal order of the diamonds. By following the above protocol, the reveal agent~$A_{z_r}$ receives the $N-1$ qubits if $N$ is even and the $N$ qubits if $N$ is odd. Hence,~$A_{z_r}$ can decode the state~$\ket{\psi}$.

   This subsection has explained our protocol for quantum summoning.  Subsection~\ref{subsec:stabilizer} provides the detail of the CSS code~(\ref{eq:NN-12}) used in our protocol.

\subsection{The CSS code}
\label{subsec:stabilizer}
In this subsection, we propose a stabilizer code with each qubit assigned to an edge of~$K_{\tilde{N}}$. 
 We show that it is an
$\left[\left[
		{\tilde{N}\choose2},
		1,\frac{\tilde{N}}{2}\right]\right]$
	CSS code, which can be used to summon a qubit by following the protocol in Subsec.~\ref{subsec:protocol}.

 \begin{theorem}
 \label{theorem:CSScode}
 The stabilizer code, specified by a $$\left[{\tilde{N}\choose2} -1\right]\times 2{\tilde{N}\choose 2}$$ stabilizer generator matrix 
 \begin{equation}
	H_{\tilde{N}}=\begin{bmatrix}[c|c]
		\bm{T}_{123} & \bm0 \\
		\bm{T}_{124} &\bm0 \\
			\vdots & \vdots \\
		\bm{T}_{1 2 \tilde{N}} & \bm0\\
		\bm{T}_{134} &\bm0 \\
			\vdots & \vdots \\
		\bm{T}_{1\, \tilde{N}-1\, \tilde{N}} & \bm0\\
		\hline
		\bm0 & \bm{A}_1+\bm{A}_2 \\
		\bm0 & \bm{A}_1+\bm{A}_3 \\
			\vdots & \vdots \\
		\bm0 & \bm{A}_1+\bm{A}_{\tilde{N}-1} \\
	\end{bmatrix},
\label{stabilizergenerator}
\end{equation}
  where $\bm0$  is an $  {\tilde{N}\choose2}$-dimensional zero vector,
  is an $\left[\left[
		{\tilde{N}\choose2},
		1,\frac{\tilde{N}}{2}\right]\right]$ CSS code, which can correct erasure errors at qubits $q_{ij}$ for $i, j\in[N]\setminus \{r\}$  for any $r\in \left[\tilde{N}\right]$. 
 \end{theorem}

 The stabilizer generator matrix~(\ref{stabilizergenerator}) is analogous to the stabilizer generator matrix of the homological continuous-variable quantum error correcting code~\cite{NJP18}. 
By changing $-1$ to~$1$ in the stabilizer generator matrix of the continuous-variable code,
one obtains the generator matrix~(\ref{stabilizergenerator}) from the generator matrix of the continuous-variable code. 
In continuous-variable codes,
$\pm 1$ in the generator matrix represents the phase-space displacement operators $\text{e}^{\pm\text{i} \hat{X}}$ or $\text{e}^{\pm\text{i} \hat{P}}$, where~$\hat{X}$
and~$\hat{P}$ are quadrature operators~\cite{PhysRevLett.88.097904,RevModPhys.84.621}. 
On the other hand,
in the qubit code,~$1$ in the generator matrix represent Pauli operators~$Z$ or~$X$.

To prove this theorem, we prove the following three lemmas.

\begin{lemma}
\label{lemma:stabilizergenerator}
$H_{\tilde{N}}$~(\ref{stabilizergenerator})  is a stabilizer generator matrix of a CSS code, which encodes one qubit into $\tilde{N}\choose2$ qubits. 
\end{lemma}

\begin{proof}  
From Eqs.~(\ref{eq:cutspace}) and~(\ref{eq:cyclespace}),  
\begin{equation}
  \bm{T}_{1jk}\cdot (\bm{A}_1+\bm{A}_l)
  	=0
\end{equation}
for any $j$, $k$, $l$ such that $2\le j<k\le \tilde{N}$ and $2\le l\le \tilde{N}-1$
 Using Eq.~(\ref{commutecondition}), we know that all the stabilizer generators in $H_{\tilde{N}}$~(\ref{stabilizergenerator}) commute with each other,
 thereby generating an Abelian subgroup $\mathcal{S}$ of $\mathcal G_{\tilde{N}\choose2}$. 
 There are ${\tilde{N}\choose2} -1$ independent stabilizer generators, so this stabilizer code encodes one qubit into~$\tilde{N}\choose2$ qubits. 
 The first 
 ${\tilde{N}-1\choose2}$ stabilizer generators contain only $Z$ operators and identities and 
  the other 
$\tilde{N}-2$ stabilizer generators contain only $X$ operators and identities.
  Thus, the stabilizer code is a CSS code. 
  \end{proof}

In $H_{\tilde{N}}$, the vectors representing the Z-type stabilizers and the X-type stabilizers span $\mathcal{C}_1^\perp$~(\ref{eq:cyclespace}) and $\mathcal{C}_2$~(\ref{eq:pairs}) for $n=\tilde{N}$ respectively. Thus, the CSS code in Theorem~\ref{theorem:CSScode} is specified by the linear codes~$\mathcal C_1$~(\ref{eq:cutspace}) and $\mathcal C_2$~(\ref{eq:pairs}) for $n=\tilde{N}$.

Now we show that by assigning each of the ${\tilde{N}\choose2}$ physical qubits to an edge in~$K_{\tilde N}$, this CSS code can correct the erasure errors 
at those qubits, which are not connected to vertex~$\diamonded{r}$, for any $r\in \left[\tilde{N}\right]$. 
Hence, by following the protocol of quantum summoning in~Subsec.~\ref{subsec:protocol}, no matter which~$A_{y_r}$ receives the request, the associated 
reveal agent $A_{z_r}$ can decode the original state~$\ket{\psi}$ from her $\tilde{N}-1$ qubits in Eq.~(\ref{accesseven}) or Eq.~(\ref{accessodd}).
  \begin{lemma}
\label{lemma:CSS1}
For any $r\in \left[\tilde{N}\right]$,  the CSS code in Theorem~\ref{theorem:CSScode} can correct erasure errors at qubits $q_{ij}$ for $i, j\in \left[\tilde{N}\right] \setminus \{r\}$.
  \end{lemma}  
  The proof of Lemma~\ref{lemma:CSS1} is in Appendix~\ref{proof:CSS}. 
  The proof is a modified version of that for the continuous-variable code~\cite{NJP18} with
  the infinite-dimensional field $\mathbb{R}$ replaced by the finite-dimensional field $\mathbb{Z}_2$.
  One side effect of this modification is that $\tilde{N}$ has to be even. 
Lemma~\ref{lemma:CSS1} is no longer true if $\tilde{N}$ is odd. To see this, we consider an example of a three-qubit code with stabilizer generators $\{ZZZ, IXX\}$. If $r=3$, this code should 
correct any Pauli  error at~$q_{12}$. This is obviously false, because the Pauli error at $q_{12}$, $ZII$, commutes with both stabilizer generators but does not lie in the stabilizer group.

 Now we find the distance of the CSS code, which is an important parameter characterizing the capability of the code to detect and correct errors.  
\begin{lemma}
The distance of the CSS code in Theorem~\ref{theorem:CSScode} is $\tilde{N}/2$.
\label{lemma:CSSdistance}
\end{lemma}

 Proof of Lemma~\ref{lemma:CSSdistance} is in Appendix~\ref{proof:CSSdistance}. Lemma~\ref{lemma:CSSdistance} implies that the CSS code can correct any $\left(\tilde{N}/2 -1\right)$-qubit erasure errors. 
 Although the distance of this CSS code scales as $O\left(\tilde{N}\right)$, Lemma~\ref{lemma:CSS1} implies that the CSS code can correct particular erasure errors at $O\left(\tilde{N}^2\right)$ qubits.

 This subsection has specified the CSS code~(\ref{eq:NN-12}) by its stabilizer generator matrix~(\ref{stabilizergenerator}). 
 We have shown the erasure errors that the CSS code, can correct and the distance of the CSS code.
 In next subsection, we explain how to encode and decode this CSS code.

 \subsection{Encoding and decoding}
 \label{subsec:encodedecode}
 
   \begin{figure}[t]
\begin{center}
\includegraphics[width=0.95\linewidth]{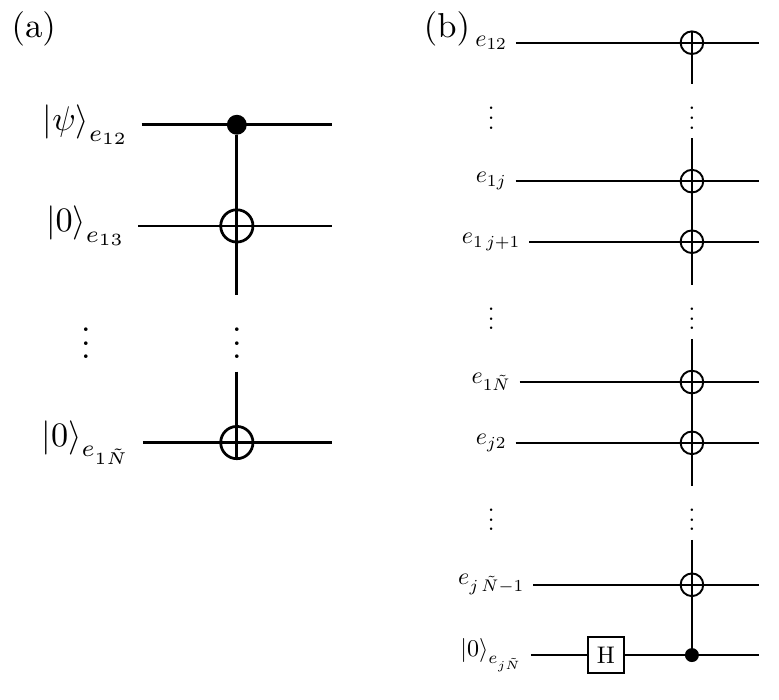}
\end{center} 
\caption{(a) Multiple CNOT gates with $\ket{\psi}_{e_{12}}$ as the control qubit and $\{\ket{0}_{e_{1 i}}\}_{i=2}^{\tilde N}$ as the target qubits; 
(b) A Hadamard gate is applied at $\ket{0}_{e_{j\tilde{N}}}$ followed by multiple CNOT gates with $\ket{0}_{e_{j \tilde N}}$ as the control qubit and the qubits assigned to $\{e_{1l}\}_{l=[\tilde N], l\neq j}\cup \{e_{j k}\}_{k=2}^{\tilde{N}-1}$ as target qubits.}
\label{fig:twoencodingcircuits}
\end{figure}

In last subsection, we have shown that the encoding of our CSS code employs~$O(N^2)$ qubits while the decoding uses only~$O(N)$ qubits.
 In this subsection, we present systematic methods to construct encoding and decoding circuits for our CSS code. 
The encoding method used here follows the standard method of encoding stabilizer codes~[21], discussed in Subsec.~II B. Our decoding method differs from the stabilizer code decoding method because it only corrects erasure errors, which occur in summoning.
 We also calculate the gate complexity of both the encoding and the decoding circuits. 
 It is shown that~$G$ in the encoding is $O(N^2)$ and~$G$ in the decoding is $O(N)$.

  To build the encoding circuit of this CSS code, we introduce logical operations on the encoded state. 
By using the vector representation~(\ref{symplecticnotation}), the logical operations are defined as 
\begin{equation}
  \bar{X}\coloneqq 
 \begin{bmatrix} \bm0 & \bm{A}_1 \end{bmatrix}
\end{equation} and 
\begin{equation}
	\bar{Z}
		\coloneqq \begin{bmatrix} \bm{A}_1 & \bm0 \end{bmatrix}.
\end{equation} 
From~$\bm{A}_1\cdot \bm{A}_1=1$ and Eq.~(\ref{commutecondition}), $\bar{X}$ and $\bar{Z}$ anti-commute with each other.

We choose
\begin{equation}
\ket{\psi_0}=\ket{\bm{0}}\coloneqq \ket{0}^{\tilde{N}\choose2},
\end{equation}
which is an eigenstate of~$\bar{Z}$ with eigenvalue one.
To encode $\ket{\psi}$~(\ref{qubitstate}), using Eqs.~(\ref{encodingzero}) and~(\ref{encodingone}), and the fact that $Z$-type stabilizers act trivially on~$\ket{\psi_0}$,
we obtain the encoded state
 \begin{widetext}
 \begin{equation}
 	 \label{CSScodebasis1}
	\alpha \ket{0}_L +\beta \ket{1}_L
	=\frac{1}{\sqrt{2^{\tilde{N}-2}}} \prod_{j=2}^{\tilde{N}-1} \left(I+\bigotimes_{i=1}^{\tilde{N}\choose2} X^{(\bm{A}_1+\bm{A}_j)_i}\right) 
	\left(\alpha \ket{\bm{0}}  +\beta \ket{\bm{A}_1} \right),
 \end{equation}
 \end{widetext}
 where 
 \begin{equation}
  \ket{\bm{A}_1}=\prod_{i=1}^{\tilde{N}\choose2} X^{(\bm{A}_1)_i}\ket{\bm{0}},
  \end{equation}  and $(\bm{A}_1+\bm{A}_j)_i$ and $(\bm{A}_1)_i$ are the $i$-th entries of vectors $\bm{A}_1+\bm{A}_j$ and~$\bm{A}_1$ respectively.
  
  \begin{figure}[t]
\begin{center}
\includegraphics[width=0.5\linewidth]{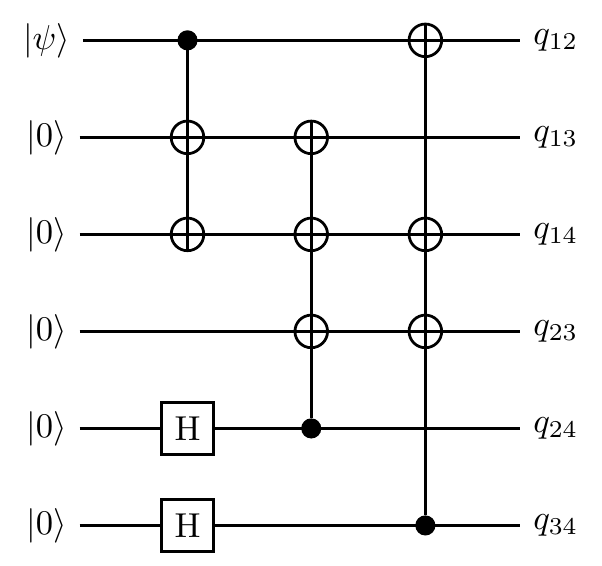}
\end{center} 
\caption{The encoding circuit of the CSS code comprising Hadamard gates and CNOT gates when $\tilde{N}=4$. The inputs of this circuit are $\ket{\psi}\otimes \ket{00000}$ and the outputs of this circuit are the qubits assigned to each edge of the complete graph~$K_4$ shown in Fig.~{\ref{fig:fourvertextournamentprism2}}(b).}
\label{fig:CSSencoding}
\end{figure}

$S$ applies CNOT gates as in Fig.~\ref{fig:twoencodingcircuits}(a) to the product state 
\begin{equation}
   \ket{\psi}\otimes \ket{0}^{{\tilde{N}\choose2}  -1},
\end{equation} 
to obtain
\begin{equation}
\label{encodingoperation}
   \alpha\ket{\bm{0}} 
   +\beta \ket{\bm{A}_1}.
\end{equation}
  Then $S$ implements each operation
  \begin{equation}
  I+\bigotimes_{i=1}^{m} X^{(\bm{A}_1+\bm{A}_j)_i} ,
\end{equation} 
 where $1\le j\le \tilde{N}-2$,
 by using CNOT gates and a Hadamard gate as in Fig.~\ref{fig:twoencodingcircuits}(b). Finally,~$S$ obtains the encoded state~(\ref{CSScodebasis1}). 
  Figure~\ref{fig:CSSencoding} presents an example of the encoding circuit for $\tilde{N}=4$.
  
 The number of CNOT gates in Fig.~\ref{fig:twoencodingcircuits}(a) is~$O(N)$.
In Fig.~\ref{fig:twoencodingcircuits}(b), the number of CNOT gates is~$O(N)$ and the number of Hadamard gate is one. 
 As the circuit in Fig.~\ref{fig:twoencodingcircuits}(a) is only applied once and the circuit in Fig.~\ref{fig:twoencodingcircuits}(b) must be applied~$O(N)$ times.
  the encoding of the CSS code consumes $O\left(N^2\right)$ CNOT gates and $O(N)$ Hadamard gates. 
  Hence, for the encoding of the CSS code,~$G \in O\left(N^2\right)$.
 
   \begin{figure}[t]
\begin{center}
\includegraphics[width=0.8\linewidth]{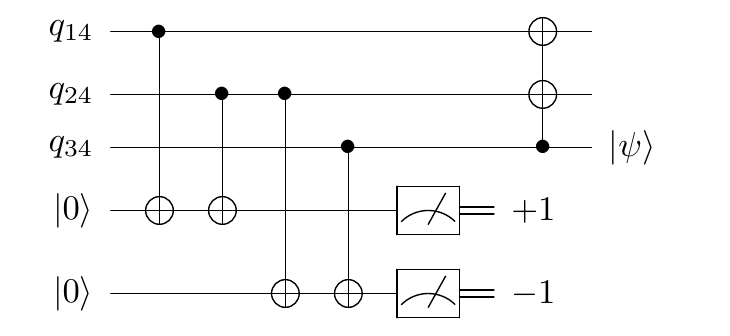}
\end{center} 
\caption{One example of the decoding circuit comprising CNOT gates and measurements of $Z$ operators for $\tilde{N}=4$. The inputs of the circuit are three physical qubits $q_{14}$, $q_{24}$ and~$q_{34}$ and two ancillary qubits $\ket{00}$. The measurement outcomes on the two ancillary qubits are $+1$ and $-1$, based on which two CNOT gates are applied with the third qubit as the control qubit and the first two qubits as the target qubits. The third output qubit is the original qubit~$\ket{\psi}$.}
\label{fig:decoding}
\end{figure}

The decoding scheme is explained in the following.
Suppose the request is sent to~$y_r$. Reveal agent~$A_{z_r}$ cannot decode by measuring the syndromes as she has only $\tilde{N}-1$ qubits.
The encoded state~(\ref{CSScodebasis1}) is an equally weighted superposition of the codewords given in Eqs.~(\ref{superpositionofcodewords0}) and~(\ref{superpositionofcodewords1}) 
(see Appendix~\ref{reducedstate}). 
After tracing out the lost qubits, the reduced state~$\rho_r$~(\ref{reduceddensitymatrix4}) becomes a mixture of the codewords.
To decode the original state~$\ket{\psi}$ from the $\tilde{N}-1$ qubits with reduced density matrix~$\rho_r$, reveal agent~$A_{z_r}$ measures the set of mutually commutative Hermitian operators 
\begin{equation}
\label{eq:measureoperator}
	\left\{Z_{q_{rk}}Z_{q_{r \, k+1}};k\in\left[\tilde{N}-1\right]\backslash\{r\}\right\},
\end{equation}
where $Z_{q_{rk}}$ represents the $Z$ operator on the qubit $q_{rk}$. 
 After applying the projective measurements, 
the reduced state is projected onto one codeword, becoming a pure state (see Appendix~\ref{reducedstate}).

 According to the measurement outcomes, 
by applying  $\tilde{N}-2$ CNOT gates with one control qubit and distinct target qubits,~$A_{z_r}$ 
obtains the original state~$\ket{\psi}$ at the control qubit.
Fig.~\ref{fig:decoding} presents an example of the decoding circuit when the request is received at $\diamonded{4}$.
In decoding, $G\in O(N)$ and the number of single-qubit measurements is also $O(N)$.

\subsection{Comparison with the CWS code}
\label{subsec:comparison}
Now we compare the quantum resources required by our CSS code with Hayden and May's CWS code. 
To investigate the complexity of the encoding of the CWS code, we need to know the complexity of preparing graph state~$\ket{G_{\text{CWS}}}$. 
From Eq.~(\ref{graphstate}), we know that the number of controlled-$Z$ gates and Hadamard gates in preparing a graph state equals to the number of edges and vertices in the graph, respectively.
The numbers of edges and vertices in~$G_{\text{CWS}}$ are
\begin{align}
|E(G_{\text{CWS}})|= &\frac{N(N-1)^2}{2},\\
 |V(G_{\text{CWS}})|= & N(N-1).
\end{align}
Thus, preparing $\ket{G_{\text{CWS}}}$ requires $O\left(N^3\right)$ controlled-$Z$ gates and $O\left(N^2\right)$ Hadamard gates.
Applying the codeword operators (\ref{codewordoperators}) requires additional $O\left(N^2\right)$ controlled-$Z$ gates. 

In conclusion, the encoding of the CWS code consumes $O\left(N^3\right)$ controlled-$Z$ gates and $O\left(N^2\right)$ Hadamard gates. 
Compared with the CWS code, our CSS code reduces $G$ for encoding from~$O\left(N^3\right)$ to~$O\left(N^2\right)$.

Our protocol employs a CSS code to summon quantum information in any valid configuration.
The CSS code can correct the erasure errors that occur in the quantum summoning task.
The encoding and the decoding methods for this CSS code have been presented. 
Finally, we have compared the complexity of the encoding of the CSS code with the encoding of the CWS code and found that our CSS code is more efficient.
  
\section{Discussion}
\label{sec:discussion}
We have presented a protocol to summon quantum information efficiently in any valid configuration of causal diamonds. Central to our protocol is a CSS code that encodes one logical qubit into $O(N^2)$ physical qubits, where each physical qubit is assigned to an edge of a complete graph whose vertices correspond to causal diamonds. 
This code is a qubit version of the homological continuous-variable quantum error correcting code~\cite{NJP18}.
The CSS code is designed using the fact that the power set of edges of a complete graph can be cast as a vector space. The stabilizer generators of the CSS code correspond to triangle graphs and sums of star graphs. 

The properties of these graphs are used to show that the logical qubit can be decoded from the subset of physical qubits that are assigned to edges adjacent to any vertex. In order to employ this code for summoning, the physical qubits are sent to the request points in such a way that the past of every reveal point contains enough physical qubits to decode the original qubit. Our protocol design, similar to one used previously~\cite{NJP18}, ensures that whenever a request agent receives the request the associated reveal agent receives all physical qubits required to decode the original qubit. 

We also present procedures to design the encoding and decoding circuits for the CSS code. We show that our protocol is less resource-intensive than the protocol based on the CWS code~\cite{JPA49} which uses circuits that are $O(N^2)$ wide and $O(N^3)$ deep. The circuits for the CSS code have width that is also $O(N^2)$ but half that of the CWS code and require only $O(N^2)$ gates.

\section{Conclusion}
\label{sec:conclusion}

Our protocol for summoning is designed to work for any valid configuration of causal diamonds, where the underlying CSS code depends only on the number of causal diamonds. It is likely that codes can be designed that reduce resource usage by exploiting the structure of causal connections between the causal diamonds, examples being when a single causal curve connects multiple causal diamonds~\cite{NJP18} or when the graph representing causal connections is acyclic~\cite{PhysRevA.93.062327}. 

While any given configuration of causal diamonds may be realized in man-made quantum networks, a useful avenue of research would be to classify the configurations that can occur naturally in flat or curved spacetimes. Our codes as well as other codes for quantum summoning assume that entangled states may be tranferred without decoherence in spacetime. Quantum summoning in curved spacetime or Rindler coordinates might require the usage of codes that protect against decoherence caused due to gravity or acceleration~\cite{Alsing2003,PhysRevLett.110.113602,PhysRevD.96.065018}.

\begin{acknowledgements}
We thank Dominic Berry, Dong-Xiao Quan, Masoud Habibi Davijani, Yun-Jiang Wang and Wei-Wei Zhang for helpful discussions and acknowledge funding from Alberta Innovates, NSERC, China's
1000 Talent Plan, and the Institute for Quantum Information and Matter, an NSF Physics Frontiers Center (NSF Grant PHY-1125565) with support of the Gordon and Betty Moore Foundation (GBMF-2644).
\end{acknowledgements}

\appendix
 
\section{Proof of Lemma~\ref{lemma:CSS1}}
\label{proof:CSS}

\begin{proof}

Denote $\mathcal E_r$ as the set of Pauli operators at qubits $q_{ij}$ for $i, j\in \left[\tilde{N}\right]\setminus \{r\}$. 
For any Pauli operator $P\in\mathcal E_r$, its vector representation~(\ref{symplecticnotation}) is denoted
\begin{equation}
P=\begin{bmatrix}  \bm{P}_Z & \bm{P}_X \end{bmatrix}.
\end{equation} 
As~$P$ acts nontrivially only at qubits $q_{ij}$ for $i, j\in \left[\tilde{N}\right] \setminus \{r\}$, 
\begin{equation}
\forall k\in\left[\tilde{N}\right]\setminus \{r\}, \; \bm{P}_Z\cdot \bm{e}_{rk}=\bm{P}_X\cdot \bm{e}_{rk}=0.
\label{errorcondition}
\end{equation}
 
 To show that the stabilizer code in Theorem~\ref{theorem:CSScode} can correct any error in $\mathcal E_r$ for every $r$, 
it is sufficient to prove that~\cite{gottesman1997stabilizer}  
\begin{equation}
\forall \, P\in \mathcal E_r, \; P\in C(\mathcal S)\Rightarrow  P\in \mathcal S,
\end{equation}
where~$\mathcal S$ is the stabilizer group generated by the stabilizer generators in~$H_{\tilde{N}}$ (\ref{stabilizergenerator}) and $C(\mathcal S)$ is the centralizer of~$\mathcal S$ in~$\mathcal G_{\tilde{N}\choose2}$, i.e.~the group of the Pauli operators commuting with all the elements of~$\mathcal S$.
Using Eq.~(\ref{commutecondition}), we know that
$P\in C(\mathcal S)$ if and only if  
\begin{equation}
\forall  \,\bm{v}\in\mathcal{C}_1^\perp, \; \bm{P}_X\cdot \bm{v}=0,
\label{commutecondition1}
\end{equation}
and 
\begin{equation}
\forall \, \bm{u}\in \mathcal{C}_2, \, \bm{P}_Z\cdot \bm{u}=0.
\label{commutecondition2}
\end{equation}

From $\bm{T}_{rij}\in \mathcal{C}_1^\perp$ and Eq.~(\ref{commutecondition1}), 
\begin{equation}
\label{PXorthogonaltoT}
\bm{P}_X\cdot \bm{T}_{rij}=0.
\end{equation}
It implies that 
\begin{equation}
	\bm{P}_X\cdot(\bm{e}_{ri}+\bm{e}_{rj}+\bm{e}_{ij})=0.
\end{equation}   
Using Eq.~(\ref{errorcondition}), we know that 
\begin{equation}
 \bm{P}_X\cdot \bm{e}_{ij}=0.
\label{PXdotEdge}
\end{equation} 
Equation~(\ref{PXdotEdge}), together with Eq.~(\ref{errorcondition}), implies that $\bm{P}_X=\bm0$.  
   
   Next we prove that Eq.~(\ref{commutecondition2}) implies that $\bm{P}_Z \in \mathcal{C}_1^\perp$.  
 Suppose 
\begin{equation}
\label{assumption}
 \bm{P}_Z\cdot \bm{A}_1=1.
\end{equation} 
 Then Eq.~(\ref{commutecondition2}) implies that for $2\le l\le \tilde{N}-1$,
\begin{equation}
 \bm{P}_Z\cdot \bm{A}_l=1.
\end{equation}
As $\tilde{N}$ is even, 
\begin{equation}
\sum_{m=1}^{\tilde{N}-1} \bm{P}_Z\cdot\bm{A}_m=1.
\end{equation}
 As 
\begin{equation}
\label{sumzero}
 \sum_{m=1}^{\tilde{N}} \bm{P}_Z\cdot\bm{A}_m= \bm{P}_Z\cdot \sum_{m=1}^{\tilde{N}} \bm{A}_m=0,
\end{equation} 
 we have
\begin{equation}
 \bm{P}_Z\cdot \bm{A}_{\tilde{N}}=1.
\end{equation} 
 Hence,
\begin{equation}
\label{allone}
 \forall\, k\in \left[\tilde{N}\right], \; \bm{P}_Z\cdot \bm{A}_k=1.
\end{equation} 
Equation~(\ref{allone}) contradicts  
 Eq.~(\ref{errorcondition}) because Eq.~(\ref{errorcondition}) implies that 
\begin{equation}
 \bm{P}_Z\cdot \bm{A}_r=0.
\end{equation}
 
 Thus, (\ref{assumption}) is false and 
\begin{equation}
 \bm{P}_Z\cdot \bm{A}_1=0.
\end{equation} 
 Then Eq.~(\ref{commutecondition2}) 
 indicates that for $2\le l\le \tilde{N}-1$,
\begin{equation}
 \bm{P}_Z\cdot \bm{A}_l=0.
\end{equation}
Hence 
\begin{equation}
\sum_{m=1}^{\tilde{N}-1} \bm{P}_Z\cdot\bm{A}_m=0.
\end{equation}
 From Eq.~(\ref{sumzero}),
  we know that
\begin{equation}
	\bm{P}_Z\cdot \bm{A}_{\tilde{N}}=0.
\end{equation} 
 Thus, 
\begin{equation}
\forall \, k\in \left[\tilde{N}\right], \;  \bm{P}_Z\cdot \bm{A}_k=0,
\end{equation}
 which indicates that $\bm{P}_Z \in\mathcal{C}_1^\perp$. Since $\bm{P}_X=\bm0$, 
  we know that $P$ is a $Z$-type stabilizer in~$\mathcal S$.

  \end{proof}

\section{Proof of Lemma~\ref{lemma:CSSdistance}}
\label{proof:CSSdistance}

\begin{proof}
 The distance (\ref{distance}) of the stabilizer code in Theorem~\ref{theorem:CSScode} equals to the minimum weight of the Pauli operators in $C(\mathcal{S})\setminus \mathcal{S}$~\cite{gottesman1997stabilizer}.
 To prove the minimum weight of the Pauli operators in $C(\mathcal{S})\setminus \mathcal{S}$ is $\tilde{N}/2$,
 we first show that there exists a Pauli operator
 \begin{equation}
 \label{errorinCSnotinS}
P\coloneqq \bigotimes_{i=1}^{\tilde{N}/2} Z_{q_{2i-1,\, 2i}}
\end{equation}
  with weight~$\tilde{N}/2$ such that $P\in C(\mathcal{S})\setminus \mathcal{S}$. 
 Then we show that no Pauli operator with weight less than~$\tilde{N}/2$ lies in $C(\mathcal{S})\setminus \mathcal{S}$.
 
 From Eq~(\ref{errorinCSnotinS}), we know that
 \begin{equation}
 \label{errorinCSnotinSvector}
 \bm{P}_Z=\sum_{i=1}^{\tilde{N}/2} \bm{e}_{2i-1,\, 2i},\; \bm{P}_X=\bm{0}.
 \end{equation}
From Eq. (\ref{errorinCSnotinSvector}), we find that 
\begin{equation}
\forall\, l \in \left[\tilde{N}\right],\; \bm{P}_Z\cdot \bm{A}_l=1.
\end{equation}
Hence, 
\begin{equation}
\forall \, \bm{v}\in \mathcal{C}_2, \,\bm{P}_Z\cdot \bm{v}=0.
\end{equation}
From the fact that $\bm{P}_X=\bm0$, we know that~$P$ (\ref{errorinCSnotinS}) commutes with all the stabilizers in~$\mathcal S$, i.e., $P\in C(\mathcal S)$. 
As $\bm{P}_Z$ cannot be represented by an Euclidean cycle, 
\begin{equation}
\bm{P}_Z \notin \mathcal{C}_1^\perp.
\end{equation}
 Thus, $E\notin \mathcal S$, and hence 
\begin{equation}
P\in C(\mathcal{S})\setminus \mathcal{S}.
\end{equation}

For any Pauli operator~$P'$ with weight less than~$\tilde{N}/2$, there exists an $r\in \left[\tilde{N}\right]$ such that 
\begin{equation}
\forall k \in \left[\tilde{N}\right]\setminus \{r\},\; \bm{P}'_Z\cdot \bm{e}_{rk}=\bm{P}'_X\cdot \bm{e}_{rk}=0.
\end{equation}
From the discussion in Appendix~\ref{proof:CSS}, we know that
\begin{equation}
P'\in C(\mathcal S)\Rightarrow  P'\in \mathcal S.
\end{equation}
It implies that
\begin{equation}
P'\notin C(\mathcal{S})\setminus \mathcal{S}.
\end{equation}
Thus, no Pauli operator with weight less than~$\tilde{N}/2$ lies in $C(\mathcal{S})\setminus \mathcal{S}$.

\end{proof} 
 
 \section{Reduced density matrix}
\label{reducedstate}

In this appendix, we calculate the reduced density matrix~$\rho_r$ and show the state after the measurements of the operators in Eq.~(\ref{eq:measureoperator}).

From Eqs.~(\ref{encodingzero}) and~(\ref{encodingone}), 
  \begin{align}
      \label{superpositionofcodewords0}
  \ket{0}_L =&\frac{1}{\sqrt{2^{\tilde{N}-2}}} \sum_{\bm{x}\in \mathcal{C}_2} \ket{\bm{x}}, \\
     \label{superpositionofcodewords1}
   \ket{1}_L =&\frac{1}{\sqrt{2^{\tilde{N}-2}}} \sum_{\bm{x}\in \mathcal{C}_2} \ket{\bm{A}_1+\bm{x}},
  \end{align}
 form a basis of the CSS code in Theorem~\ref{theorem:CSScode}.
Hence, the density matrix of the  ${\tilde{N}\choose2}$ physical qubits encoding~$\ket{\psi}$~(\ref{qubitstate}) is
\begin{align}
	\rho
		=&\frac{1}{2^{\tilde{N}-2}}\sum_{\bm{x}\in \mathcal{C}_2} \left( \alpha \ket{\bm{x}}+\beta \ket{\bm{A}_1+\bm{x}}\right)
			\nonumber\\
		&\times\sum_{\bm{y}\in\mathcal{C}_2}  \left(\alpha^* \bra{\bm{y}}+\beta^*\bra{\bm{A}_1+\bm{y}}|\right).
\end{align}

To express the reduced density matrix, we give the following notations. 
The subset of edges connected to~$\diamonded{r}$ in $K_{\tilde N}$  is denoted by $E_r$ and the complement of $E_r$ in $E$, 
i.e.~the subset of edges not connected to~$\diamonded{r}$, is denoted by $E_r^c$. 
In the same way as~$2^{E_K}$ forms a linear space~$$\mathscr{E}\cong \mathbb{Z}_2^{\tilde{N}\choose2},$$ the power set  $2^{E_r}$ forms a linear subspace
\begin{equation}
\mathscr{E}_r\cong \mathbb{Z}_2^{\tilde{N}-1}
\end{equation} and 
the power set $2^{E_r^c}$ forms the orthogonal complement of~$\mathscr{E}_r$ in~$\mathscr{E}$, denoted by 
\begin{equation}
\mathscr{E}_r^{\perp}\cong \mathbb{Z}_2^{\tilde{N}-1\choose2}.
\end{equation}
 For a vector $\bm{v} \in\mathscr E$, 
 we use $\bm{v}_r$ to denote the projection of~$\bm{v}$ onto $\mathscr{E}_r$, and
 $\bm{v}_r^{\perp}$ to denote the projection of~$\bm{v}$ onto $\mathscr{E}_r^{\perp}$.

 Now we calculate the reduced density matrix $\rho_r$ of the $\tilde{N}-1$ qubits  $\left \{q_{rk}; k\in \left[\tilde{N}\right]\setminus \{r\} \right \}$.
  As  
\begin{equation}
\forall r\in \left[\tilde{N}\right],\; \sum_{\bm{x}\in \mathcal{C}_2} \ket{\bm{A}_r+\bm{x}}= \sum_{\bm{x}\in \mathcal{C}_2} \ket{ \bm{A}_1+\bm{x}},
\end{equation} 
\begin{widetext}
 \begin{equation}
\label{reduceddensitymatrix1}
 \rho_r =\frac{1}{2^{\tilde{N}-2}}\operatorname{tr}_{E^c_r} \left[ \sum_{\bm{x}\in \mathcal{C}_2} \Big(\alpha \ket{\bm{x}}+\beta \ket{\bm{A}_r+\bm{x}}\Big) 
  \sum_{\bm{y}\in \mathcal{C}_2}  \Big(\alpha^* \langle \bm{y}|+\beta^* \langle \bm{A}_r+\bm{y}|\Big) \right],
 \end{equation}
 where $\operatorname{tr}_{E_r^c}$ denotes the partial trace over the qubits $\left\{ q_{ij}; \, i, j \in \left[ \tilde{N} \right] \setminus \{r\} \right\}$.
 From the definition of partial trace~\cite{nielsen2000quantum}, 
 \begin{equation}
\label{reduceddensitymatrix2}
   \rho_r =\frac{1}{2^{\tilde{N}-2}}\sum_{\bm{v}\in  \mathscr{E}_r^{\perp}} \left\langle \bm{v} \left| \sum_{\bm{x}\in \mathcal{C}_2} \Big(\alpha \ket{\bm{x}}+\beta \ket{\bm{A}_r+\bm{x}}\Big)  
  \sum_{\bm{y}\in \mathcal{C}_2} \Big(\alpha^* \langle \bm{y}|+\beta^* \langle \bm{A}_r+\bm{y}|\Big) \right |\bm{v}\right \rangle. 
  \end{equation}
For each $\bm{v}\in \mathscr{E}_r^{\perp}$, there is at most one $\bm{x}\in \mathcal{C}_2$ such that
\begin{equation}
 \bm{x}_r^{\perp}=\bm{v}.
\end{equation}
Hence, we get
  \begin{equation}
\label{reduceddensitymatrix3}
  \rho_r  =\frac{1}{2^{\tilde{N}-2}}\sum_{\bm{x}\in \mathcal{C}_2} \left \langle \bm{x}_r^{\perp} \right|  \left(\alpha \ket{\bm{x}}+\beta\ket{\bm{A}_r +\bm{x}}\right) \left(\alpha^* \langle \bm{x} |+\beta^* \langle \bm{A}_r+\bm{x} |\right) \left|\bm{x}_r^{\perp}\right\rangle. 
  \end{equation}
  As 
  \begin{equation}
 \langle \bm{x}_r^{\perp}\ket{\bm{x}}=\ket{\bm{x}_r} \text{ and } \langle \bm{x}_r^{\perp}\ket{\bm{A}_r +\bm{x}}=\ket{\bm{1} +\bm{x}_r},
 \end{equation}
 where $\bm{1}$ is an ($\tilde{N}-1$)-dimensional vector with all the entries equal to $1$,
  \begin{equation}
\label{reduceddensitymatrix4}
 \rho_r  =\frac{1}{2^{\tilde{N}-2}}\sum_{\bm{x}\in \mathcal{C}_2} \left(\alpha \ket{\bm{x}_r}+\beta \ket{\bm{1} +\bm{x}_r}\right) \left(\alpha^* \langle \bm{x}_r|+\beta^* \langle \bm{1}+\bm{x}_r|\right).
\end{equation}
 \end{widetext}

Consider the set of mutually commutative Hermitian operators in~(\ref{eq:measureoperator})
\begin{equation*}
\left\{Z_{q_{rk}}Z_{q_{r \, k+1}};k\in\left[\tilde{N}-1\right]\backslash\{r\}\right\}.
\end{equation*}
$\forall \bm{x}\in \mathcal{C}_2$, $\ket{\bm{x}_r}$ and $\ket{\bm{1} +\bm{x}_r}$ are the eigenstates of each Hermitian operator in~(\ref{eq:measureoperator}) with same eigenvalue,
so any linear combination $$\alpha \ket{\bm{x}_r}+\beta \ket{\bm{1} +\bm{x}_r}$$ is a common eigenstate of the Hermitian operators~(\ref{eq:measureoperator}), 
with the eigenvalues forming a vector consisting of $\pm 1$.
For 
\begin{equation}
\label{conditionxz}
\forall \bm{x}, \bm{z} \in \mathcal{C}_2 \text{ and } \bm{x}\neq \bm{z},
\end{equation} 
the two eigenstates
\begin{equation}\label{code1}
\alpha \ket{\bm{x}_r}+\beta \ket{\bm{1} +\bm{x}_r}
\end{equation} and
\begin{equation}\label{code2}
\alpha \ket{\bm{z}_r}+\beta \ket{\bm{1} +\bm{z}_r}
\end{equation} 
have different eigenvalue vectors.
This is because if~(\ref{code1}) and~(\ref{code2}) have the same eigenvalues,
then either $\bm{x}_r=\bm{z}_r$ or $\bm{x}_r=\bm{z}_r+\bm{1}$, both of which contradict condition~(\ref{conditionxz}).
Thus, $\rho_r$ in Eq.~(\ref{reduceddensitymatrix4}) is an equally weighted mixture of the common eigenstates of the Hermitian operators~(\ref{eq:measureoperator}) with different eigenvalue vectors.

After the projective measurements on all the Hermitian operators~(\ref{eq:measureoperator}),
the reduced state is projected onto 
\begin{equation}
\label{stateaftermeasurement}
\alpha \ket{\bm{y}_r}+\beta \ket{\bm{1} +\bm{y}_r},
\end{equation}
where $\bm{y}\in \mathcal{C}_2$ and $\bm{y}_r$ is the projection of $\bm{y}$ onto~$\mathscr E_r$. 
The corresponding measurement outcomes are $\left\{(-1)^{\bm{y}_r^{(i)}+ \bm{y}_r^{(i+1)}}; i\in \left[\tilde{N}-2\right] \right\}$,
where $\bm{y}_r^{(i)}$ is the $i$-th component in~$\bm{y}_r$.
(\ref{stateaftermeasurement}) is the state after the measurements of the Hermitian operators in~(\ref{eq:measureoperator}).

\bibliography{ref}
\end{document}